\documentclass[11pt]{article}

\usepackage[margin=1in]{geometry}
\usepackage{amsmath,amsthm,amssymb,mathrsfs,mathtools}
\usepackage{tabularx}
\usepackage{microtype}
\usepackage{booktabs}
\usepackage{enumitem}
\usepackage{xcolor}
\usepackage[unicode,breaklinks=true]{hyperref}
\usepackage{graphicx}
\usepackage{placeins}
\usepackage{float}

\usepackage{algorithm}
\usepackage{algpseudocode}

\usepackage{tikz}
\usetikzlibrary{arrows.meta,positioning,shapes.geometric}

\usepackage{listings}
\newtheorem{theorem}{Theorem}[section]
\newtheorem{lemma}[theorem]{Lemma}

\newtheorem{definition}[theorem]{Definition}
\newtheorem{remark}[theorem]{Remark}

\newcommand{\calL}{\mathcal{L}}
\newcommand{\calG}{\mathcal{G}}
\newcommand{\calH}{\mathcal{H}}
\newcommand{\calR}{\mathcal{R}}

\newcommand{\eps}{\varepsilon}
\newcommand{\Tr}{\operatorname{tr}}
\newcommand{\Der}{\Rightarrow^{*}}

\newcommand{\BMM}{\textsc{BMM}}
\newcommand{\OMv}{\textsc{OMv}}
\newcommand{\Dyck}{\textsc{Dyck}}

\newcommand{\bigO}{\mathcal{O}}

\title{%
Grammar-Constrained (CFL) Reachability:\\
Subcubic Preprocessing, Indexing Trade-offs,\\
and Structured Decoding Semantics%
}
\author{%
Faruk Alpay\thanks{Department of Computer Engineering, Bahcesehir University. \texttt{faruk.alpay@bahcesehir.edu.tr}}%
\and
Levent Sarioglu\thanks{Department of Computer Engineering, Bahcesehir University. \texttt{levent.sarioglu@bahcesehir.edu.tr}}%
}
\date{}

\begin{document}
\maketitle

\begin{abstract}
We study grammar-constrained reachability (CFL reachability / context-free path queries) in directed edge-labeled graphs, where a pair $(s,t)$ is accepted iff there exists a path whose trace belongs to $\calL_{\calG}(S)$ for a fixed context-free grammar $\calG$.
We tighten the baseline cubic-time saturation index for Chomsky normal form (CNF) grammars and then add four extensions:
(i) a nontrivial subcubic preprocessing class (linear grammars on sparse graphs) with explicit proof and index construction;
(ii) an alternative index that changes trade-offs via \emph{shared witness DAG compression} and optional dynamic support on a restricted class (bidirected Dyck reachability);
(iii) a new decomposition lemma (terminal-anchored normalization plus propagation accounting) used to derive the subcubic bound; and
(iv) a formal application shift to \emph{single-path semantics} (shortest accepted path), which changes the algorithmic objective and motivates a distance-aware index for structured decoding.
We validate the practical relevance of the linear grammar class empirically: an analysis of $9{,}558$ real-world JSON schemas from JSONSchemaBench~\cite{jsonschemabench2025} shows that $8.4\%$ fall into the strictly linear class under our conversion, with non-linearity overwhelmingly attributable to variable-length array types---a finding that informs when the subcubic regime applies in grammar-constrained neural decoding.
\end{abstract}

\section*{Executive Summary}
This manuscript presents a concise algorithms/data-structures treatment of validation-gated reachability (CFL reachability / context-free path queries).
It (a) states the problem and the standard saturation index precisely under a CNF assumption, (b) establishes a provably better preprocessing bound for a meaningful grammar/graph class (linear CFGs on sparse graphs), (c) defines a second index with different trade-offs (compressed witnesses and optional dynamic edge updates on a restricted Dyck/bidirected class), (d) introduces a formal application variant (single-path semantics: shortest accepted path) that changes the algorithmic objective and motivates a distance-aware index, and (e) provides an empirical grammar-class census over $9{,}558$ real-world JSON schemas to quantify the prevalence of the linear subclass in practice.
Side-by-side tables of preprocessing/space/query/witness/update costs, pseudocode for all indices, and proof sketches for the key claims are provided.

\section{Problem Statement and Preliminaries}

\subsection{Labeled graphs}
\begin{definition}[Labeled directed graph]
A labeled directed graph is $G=(V,E,\Sigma,\lambda)$ where $V$ is a finite set of vertices, $E\subseteq V\times V$ is a set of directed edges, $\Sigma$ is a finite alphabet, and $\lambda:E\to\Sigma$ labels edges.
Let $n\stackrel{\text{def}}{=} |V|$ and $m\stackrel{\text{def}}{=}|E|$.
\end{definition}

A path $\pi$ is a sequence of edges $(v_0,v_1)\cdots(v_{k-1},v_k)$.
Its \emph{trace} is $\Tr(\pi)=\lambda(v_0,v_1)\cdots\lambda(v_{k-1},v_k)\in\Sigma^*$.

\subsection{CFGs and CNF}
\begin{definition}[CFG]
A context-free grammar is $\calG=(N,\Sigma,P,S)$ where $N$ is a finite set of nonterminals, $\Sigma$ terminals, $P$ productions, and $S\in N$ start symbol.
Write $\calL_{\calG}(A)$ for the language generated by $A$.
Let $|N|$ be the number of nonterminals and $|P|$ the number of productions.
\end{definition}

\begin{definition}[CNF]
A CFG is in Chomsky normal form if each production is $A\to BC$, $A\to a$, or (optionally) $S\to\eps$, with $A,B,C\in N$ and $a\in\Sigma$.
\end{definition}

\subsection{Validation-gated reachability}
\begin{definition}[Validation-gated instance and query]
A validation-gated instance is $\calH=(G,\calG)$ with common alphabet $\Sigma$.
For $s,t\in V$, define the (boolean) query
\[
\mathrm{CFR}_{\calH}(s,t) \;=\; \exists \text{ path }\pi:s\leadsto t \text{ such that }\Tr(\pi)\in\calL_{\calG}(S).
\]
\end{definition}

\begin{definition}[Nonterminal relations]
For each $A\in N$, define $\calR_A\subseteq V\times V$ by
\[
(u,v)\in\calR_A \iff \exists \pi:u\leadsto v \text{ with }\Tr(\pi)\in\calL_{\calG}(A).
\]
Then $\mathrm{CFR}_{\calH}(s,t)$ holds iff $(s,t)\in\calR_S$.
\end{definition}

\paragraph{Baseline objective.}
Build an \emph{index} supporting $O(1)$ membership queries for $\calR_S$ and returning a witness path (explicit or compressed).

\section{Baseline Saturation Index for CNF Grammars}

\subsection{Index and witnesses}
\begin{definition}[Saturation index \textsc{SatIndex}]
Assume $\calG$ is CNF.
\textsc{SatIndex} stores:
(i) Boolean matrices $M_A\in\{0,1\}^{n\times n}$ encoding $\calR_A$;
(ii) a witness table $W[A,u,v]$ for each true entry, sufficient to reconstruct a witness.
\end{definition}

Witness records use three forms:
\[
\langle \textsf{term},(u,v)\rangle,\quad
\langle \textsf{eps}\rangle,\quad
\langle \textsf{bin},B,C,m\rangle
\]
meaning $A\to a$ via edge $(u,v)$; $S\to\eps$ on the diagonal; or $A\to BC$ with split vertex $m$.

\subsection{Construction}
Algorithm~\ref{alg:satindex} is the standard worklist saturation (as in classic CFL-reachability formulations).

\begin{algorithm}[H]
\caption{\textsc{SatIndex-Build}$(G,\calG)$ (CNF)}
\label{alg:satindex}
\begin{algorithmic}[1]
\Require $G=(V,E,\Sigma,\lambda)$, $\calG=(N,\Sigma,P,S)$ in CNF
\Ensure $\{M_A\}_{A\in N}$ and witnesses $W$
\State Initialize all $M_A\gets 0$ and $W\gets\bot$; queue $Q\gets\emptyset$
\ForAll{$A\to a\in P$}
  \ForAll{$(u,v)\in E$ with $\lambda(u,v)=a$}
    \If{$M_A[u,v]=0$}
      \State $M_A[u,v]\gets 1$; $W[A,u,v]\gets\langle\textsf{term},(u,v)\rangle$
      \State Enqueue $(A,u,v)$
    \EndIf
  \EndFor
\EndFor
\If{$S\to\eps\in P$}
  \ForAll{$u\in V$}
    \If{$M_S[u,u]=0$}
      \State $M_S[u,u]\gets 1$; $W[S,u,u]\gets\langle\textsf{eps}\rangle$
      \State Enqueue $(S,u,u)$
    \EndIf
  \EndFor
\EndIf
\While{$Q\neq\emptyset$}
  \State Dequeue $(X,i,j)$
  \ForAll{$A\to BC\in P$}
    \If{$X=B$}
      \ForAll{$k\in V$ with $M_C[j,k]=1$}
        \If{$M_A[i,k]=0$} \State $M_A[i,k]\gets1$; $W[A,i,k]\gets\langle\textsf{bin},B,C,j\rangle$; Enqueue $(A,i,k)$ \EndIf
      \EndFor
    \EndIf
    \If{$X=C$}
      \ForAll{$h\in V$ with $M_B[h,i]=1$}
        \If{$M_A[h,j]=0$} \State $M_A[h,j]\gets1$; $W[A,h,j]\gets\langle\textsf{bin},B,C,i\rangle$; Enqueue $(A,h,j)$ \EndIf
      \EndFor
    \EndIf
  \EndFor
\EndWhile
\State \Return $M,W$
\end{algorithmic}
\end{algorithm}

\subsection{Correctness and complexity}
\begin{theorem}[Correctness of \textsc{SatIndex}]
\label{thm:sat-correct}
After Algorithm~\ref{alg:satindex}, for all $A\in N$ and $u,v\in V$,
\[
M_A[u,v]=1 \iff (u,v)\in\calR_A.
\]
\end{theorem}
\begin{proof}
We prove soundness and completeness.

\emph{Soundness.} Consider the first time an entry $M_A[u,v]$ is set to $1$.
If it is created by a rule $A\to a$ and an edge $(u,v)$ with $\lambda(u,v)=a$, then the length-one path $(u,v)$ witnesses $(u,v)\in\calR_A$.
If it is created by the optional rule $S\to\eps$ on the diagonal, then the empty path at $u$ witnesses $(u,u)\in\calR_S$.
Otherwise, $M_A[u,v]$ is created by some rule $A\to BC$ and a split vertex $m$ with $M_B[u,m]=1$ and $M_C[m,v]=1$ already present.
By induction on the time of insertion, there exist paths $\pi_1:u\leadsto m$ and $\pi_2:m\leadsto v$ with traces in $\calL_{\calG}(B)$ and $\calL_{\calG}(C)$, respectively.
Concatenating yields a path $\pi=\pi_1\cdot\pi_2$ from $u$ to $v$ whose trace lies in $\calL_{\calG}(A)$ by the rule $A\to BC$.
Thus every true entry corresponds to some valid witness path, so $M_A[u,v]=1$ implies $(u,v)\in\calR_A$.

\emph{Completeness.} Suppose $(u,v)\in\calR_A$.
Then there is a path $\pi:u=v_0\leadsto v_k=v$ with trace $w\in\calL_{\calG}(A)$.
Fix a CNF derivation tree for $A\Der w$ and induct on its height.
If the height is $1$, then either $A\to a$ with $w=a$ or (only for $A=S$) $A\to\eps$ with $w=\eps$; the initialization step inserts the corresponding entry.
Otherwise, the root rule is $A\to BC$ and $w=w_1w_2$ where $B\Der w_1$ and $C\Der w_2$.
Let $|w_1|=\ell$ and let $m=v_\ell$ be the vertex on $\pi$ after the first $\ell$ edges; then the prefix path $u\leadsto m$ has trace $w_1$ and the suffix path $m\leadsto v$ has trace $w_2$.
By the induction hypothesis, the algorithm eventually inserts $M_B[u,m]=1$ and $M_C[m,v]=1$.
When the later of these two entries is dequeued, the rule $A\to BC$ is applied and $M_A[u,v]$ is inserted.
Since the queue runs to exhaustion, every such $(u,v)$ is eventually added, so $(u,v)\in\calR_A$ implies $M_A[u,v]=1$.

Combining soundness and completeness proves the claim.
\end{proof}

\begin{theorem}[Baseline preprocessing/space/query bounds]
\label{thm:sat-bounds}
Algorithm~\ref{alg:satindex} runs in $\bigO(|P|\cdot n^3)$ time and uses $\bigO(|N|\cdot n^2)$ space for matrices plus $\bigO(|N|\cdot n^2)$ witness records.
Queries are answered in $O(1)$ by reading $M_S[s,t]$.
\end{theorem}
\begin{proof}
Let $B$ be the set of binary productions $A\to BC$.
The initialization over terminal rules inserts at most one entry per pair $(A,u,v)$, so it costs $O(|P|m)$ time by scanning labeled edges (or $O(m)$ after grouping edges by label), and the optional $S\to\eps$ initialization costs $O(n)$.
Both are dominated by the main saturation loop.

Each queue element $(X,i,j)$ corresponds to a newly inserted matrix entry, hence there are at most $|N|n^2$ dequeues.
Fix a binary rule $A\to BC\in B$.
Consider the set of triples
\[
T_{A\to BC}=\{(i,j,k)\in V^3 : M_B[i,j]=1 \text{ and } M_C[j,k]=1\}.
\]
When $(B,i,j)$ is dequeued, the algorithm iterates over all $k$ with $M_C[j,k]=1$ and inspects exactly the triples in $T_{A\to BC}$ whose left part is $(i,j)$; when $(C,j,k)$ is dequeued, it iterates over all $h$ with $M_B[h,j]=1$ and inspects the triples whose right part is $(j,k)$.
Thus each triple in $T_{A\to BC}$ is processed at most twice (once from the left and once from the right).
Since $|T_{A\to BC}|\le n^3$, the total work charged to rule $A\to BC$ is $O(n^3)$, and summing over all binary rules yields $O(|P|n^3)$ time.
The constant-time checks/updates ($M_A$ and witness insertion) are subsumed in this bound.

For space, each $A\in N$ stores a boolean matrix $M_A\in\{0,1\}^{n\times n}$ and a constant-size witness per true entry, for $O(|N|n^2)$ space each; the queue and any row/column lists fit within the same asymptotic bound.
Finally, queries are answered in $O(1)$ by reading $M_S[s,t]$.
\end{proof}

\subsection{Witness extraction}
\begin{algorithm}[H]
\caption{\textsc{ExtractPath}$(A,u,v,W)$ (explicit path)}
\label{alg:extract}
\begin{algorithmic}[1]
\Require $W[A,u,v]\neq\bot$
\If{$W[A,u,v]=\langle\textsf{term},(u,v)\rangle$} \State \Return the single edge $(u,v)$ \EndIf
\If{$W[A,u,v]=\langle\textsf{eps}\rangle$} \State \Return empty path at $u$ \EndIf
\If{$W[A,u,v]=\langle\textsf{bin},B,C,m\rangle$}
  \State $\pi_1\gets\textsc{ExtractPath}(B,u,m,W)$
  \State $\pi_2\gets\textsc{ExtractPath}(C,m,v,W)$
  \State \Return $\pi_1\cdot \pi_2$
\EndIf
\end{algorithmic}
\end{algorithm}

\begin{remark}[Witness size]
The explicit witness length can be $\Theta(m)$ or larger depending on the graph; the index guarantees only output-linear extraction in the explicit representation.
Section~\ref{sec:swd} gives compressed witnesses with different trade-offs.
\end{remark}

\section{A Subcubic Preprocessing Class: Linear Grammars on Sparse Graphs}

This section implements extension (1): a nontrivial grammar/graph class with provably subcubic preprocessing, assuming sparse graphs $m=\bigO(n)$.

\subsection{Linear grammars and terminal-anchored normal form}
CNF is not a linear-normal-form; therefore I separate the \emph{baseline CNF assumption} (Sections 1--2) from the \emph{linear-grammar subclass} used for the improved bound.

\begin{definition}[Linear CFG]
A CFG is \emph{linear} if each production has at most one nonterminal on the right-hand side.
Equivalently, productions are of the form $A\to xBy$ or $A\to x$ where $x,y\in\Sigma^*$ and $A,B\in N$.
\end{definition}

\begin{definition}[Terminal-anchored linear normal form (TALNF)]
A linear grammar is in TALNF if each production is one of:
\[
A\to a,\quad A\to aB,\quad A\to Ba,\quad S\to\eps,
\]
where $a\in\Sigma$ and $A,B\in N$.
\end{definition}

\begin{lemma}[Terminal-anchored normalization]
\label{lem:talnf}
Every linear CFG $\calG$ can be transformed into an equivalent TALNF grammar $\calG'$ in time $\bigO(|P|)$ and size $\bigO(|P|)$, preserving the generated language.
\end{lemma}
\begin{proof}
Write productions in the linear form $A\to xBy$ or $A\to x$ with $x,y\in\Sigma^*$.
First eliminate $\eps$-productions except possibly $S\to\eps$, and remove unit rules $A\to B$ by the standard unit-closure procedure; for linear grammars this preserves linearity and takes $O(|P|)$ time with at most constant-factor size increase.
Henceforth, no $A\to B$ rules remain and only $S\to\eps$ may persist.

We now transform each remaining production independently.
Let $x=a_1\cdots a_p$ and $y=b_1\cdots b_q$.

\emph{Terminal-only rules.}
If $A\to a$ then keep it.
If $A\to a_1\cdots a_p$ with $p\ge 2$, introduce fresh nonterminals $T_1,\dots,T_{p-1}$ and replace it by
\[
A\to a_1T_1,\quad
T_i\to a_{i+1}T_{i+1}\ (1\le i\le p-2),\quad
T_{p-1}\to a_p.
\]
These are all of the form $A\to aB$ or $A\to a$ and generate exactly $x$.

\emph{Mixed rules $A\to xBy$.}
If $q=0$ and $p\ge 1$, create a prefix chain ending in $B$:
\[
A\to a_1A_1,\quad
A_i\to a_{i+1}A_{i+1}\ (1\le i\le p-2),\quad
A_{p-1}\to a_pB.
\]
If $q\ge 1$, define new nonterminals $B^{(1)},\dots,B^{(q)}$ with
\[
B^{(1)}\to Bb_1,\quad
B^{(i)}\to B^{(i-1)}b_i\ (2\le i\le q).
\]
Then $B^{(q)}\Der L(B)\,b_1\cdots b_q$ using only $A\to Ba$ rules.
If $p\ge 1$, replace $A\to xBy$ by the same prefix chain as above but ending in $B^{(q)}$ (all rules are $A\to aB$).
If $p=0$ (so the rule is $A\to By$), use the suffix chain with $A$ as the final node:
for $q=1$, set $A\to Bb_1$; for $q\ge 2$, introduce $B^{(1)},\dots,B^{(q-1)}$ and add
\[
B^{(1)}\to Bb_1,\quad
B^{(i)}\to B^{(i-1)}b_i\ (2\le i\le q-1),\quad
A\to B^{(q-1)}b_q.
\]
The remaining case $p=q=0$ would be a unit rule $A\to B$, already removed.

\emph{Equivalence.}
Each construction emits the prefix $x$, threads the unique nonterminal $B$, and then appends the suffix $y$ in the correct order; no other strings are introduced.
Conversely, any derivation through the new rules collapses to the original production.
Thus $L_{\calG}(A)=L_{\calG'}(A)$ for all $A$.

\emph{Size and time.}
For a production with $|x|=p$ and $|y|=q$, we add $O(p+q)$ new rules and nonterminals.
Summed over all productions, this is $O(|P|)$ in the standard size measure (total right-hand-side length), and the construction runs in linear time.
\end{proof}

\subsection{Linear reachability index \textsc{LinIndex}}
\begin{definition}[\textsc{LinIndex}]
Assume $\calG$ is in TALNF.
\textsc{LinIndex} stores the same $\{M_A\}$ relations as \textsc{SatIndex} but uses only \emph{terminal-anchored} propagation rules, yielding $\bigO(|P|\cdot m\cdot n)$ preprocessing.
\end{definition}

To implement propagation efficiently, precompute label-partitioned adjacency:
\[
\mathrm{In}_a(v)=\{u\in V : (u,v)\in E\text{ and }\lambda(u,v)=a\},\quad
\mathrm{Out}_a(u)=\{v\in V : (u,v)\in E\text{ and }\lambda(u,v)=a\}.
\]

Algorithm~\ref{alg:linindex} iterates over these adjacency lists to realize terminal-anchored propagation.
Figure~\ref{fig:lin-rule} illustrates the $A\to aB$ case; the $A\to Ba$ case is symmetric.

\begin{figure}[H]
\centering
\begin{tikzpicture}[>=Latex, node distance=1.8cm]
\node[circle,draw,minimum size=7mm] (u) {$u$};
\node[circle,draw,minimum size=7mm,right=of u] (x) {$x$};
\node[circle,draw,minimum size=7mm,right=of x] (v) {$v$};

\draw[->] (u) -- node[above] {$a$} (x);
\draw[->,dashed,thick] (x) to[bend left=20] node[above] {$B$} (v);
\draw[->,very thick] (u) to[bend right=20] node[below] {$A$} (v);

\node[below=0.9cm of x,align=center] (cap) {\small Rule $A\to aB$:\\
{\small if $u\xrightarrow{a}x$ and $x\xRightarrow{B}v$ then $u\xRightarrow{A}v$.}};
\end{tikzpicture}
\caption{Terminal-anchored inference used by \textsc{LinIndex}. Dashed arrows are derived relations.}
\label{fig:lin-rule}
\end{figure}

The two terminal-anchored inference patterns can be stated explicitly. For each $A\in N$,
\[
\frac{A\to aB\in P,\ (u,x)\in E,\ \lambda(u,x)=a,\ (x,v)\in\calR_B}{(u,v)\in\calR_A}
\qquad
\frac{A\to Ba\in P,\ (u,x)\in\calR_B,\ (x,v)\in E,\ \lambda(x,v)=a}{(u,v)\in\calR_A}.
\]
Algorithm~\ref{alg:linindex} realizes these rules by scanning $\mathrm{Out}_a(u)$ or $\mathrm{In}_a(v)$, respectively, and inserting newly discovered pairs into the worklist.

\begin{algorithm}[H]
\caption{\textsc{LinIndex-Build}$(G,\calG)$ (TALNF)}
\label{alg:linindex}
\begin{algorithmic}[1]
\Require $G=(V,E,\Sigma,\lambda)$, linear $\calG$ in TALNF
\Ensure $\{M_A\}_{A\in N}$ and witnesses $W$
\State Initialize all $M_A\gets 0$ and $W\gets\bot$; queue $Q\gets\emptyset$
\Comment{Terminal rules}
\ForAll{$A\to a\in P$}
  \ForAll{$(u,v)\in E$ with $\lambda(u,v)=a$}
    \If{$M_A[u,v]=0$}
      \State $M_A[u,v]\gets 1$; $W[A,u,v]\gets\langle\textsf{term},(u,v)\rangle$; Enqueue $(A,u,v)$
    \EndIf
  \EndFor
\EndFor
\If{$S\to\eps\in P$}
  \ForAll{$u\in V$}
    \If{$M_S[u,u]=0$}
      \State $M_S[u,u]\gets 1$; $W[S,u,u]\gets\langle\textsf{eps}\rangle$; Enqueue $(S,u,u)$
    \EndIf
  \EndFor
\EndIf
\While{$Q\neq\emptyset$}
  \State Dequeue $(B,x,v)$
  \Comment{Rules of form $A\to aB$}
  \ForAll{productions $A\to aB\in P$}
    \ForAll{$u\in \mathrm{In}_a(x)$}
      \If{$M_A[u,v]=0$}
        \State $M_A[u,v]\gets 1$; $W[A,u,v]\gets\langle\textsf{linL},a,(u,x),B\rangle$; Enqueue $(A,u,v)$
      \EndIf
    \EndFor
  \EndFor
  \Comment{Rules of form $A\to Ba$}
  \ForAll{productions $A\to Ba\in P$}
    \ForAll{$w\in \mathrm{Out}_a(v)$}
      \If{$M_A[x,w]=0$}
        \State $M_A[x,w]\gets 1$; $W[A,x,w]\gets\langle\textsf{linR},B,a,(v,w)\rangle$; Enqueue $(A,x,w)$
      \EndIf
    \EndFor
  \EndFor
\EndWhile
\State \Return $M,W$
\end{algorithmic}
\end{algorithm}

\begin{theorem}[Correctness of \textsc{LinIndex}]
\label{thm:lin-correct}
Assume $\calG$ is in TALNF. Algorithm~\ref{alg:linindex} computes exact $\calR_A$ relations:
$M_A[u,v]=1 \iff (u,v)\in\calR_A$ for all $A,u,v$.
\end{theorem}
\begin{proof}[Proof sketch]
The proof mirrors Theorem~\ref{thm:sat-correct} but uses TALNF derivations instead of CNF.
Each TALNF step either emits one terminal edge ($A\to a$), prepends a terminal edge ($A\to aB$), appends a terminal edge ($A\to Ba$), or contributes $\eps$ for $S$.
The worklist enumerates exactly the closure under these three inference patterns.
\end{proof}

\begin{theorem}[Subcubic preprocessing for sparse graphs]
\label{thm:lin-subcubic}
Let $\calG$ be in TALNF. Algorithm~\ref{alg:linindex} runs in time $\bigO(|P|\cdot m\cdot n)$ and uses $\bigO(|N|\cdot n^2)$ space.
In particular, on sparse graphs with $m=\bigO(n)$, preprocessing is $\bigO(|P|\cdot n^2)$, which is subcubic in $n$.
\end{theorem}
\begin{proof}
We account for the total number of iterations of the innermost loops.
Initialization over terminal rules $A\to a$ scans labeled edges once and inserts each matrix entry at most once, so it costs $O(|P|\cdot m)$ time (or $O(m)$ after grouping edges by label); the optional $S\to\eps$ initialization costs $O(n)$.
These terms are dominated by the main worklist.

Each queue element $(B,x,v)$ corresponds to a newly inserted true entry $M_B[x,v]=1$, so the number of dequeues is at most $|N|\cdot n^2$.
Fix a production $A\to aB$.
Whenever a pair $(B,x,v)$ is dequeued, the algorithm scans $\mathrm{In}_a(x)$ and performs $O(|\mathrm{In}_a(x)|)$ membership checks and at most that many insertions for $A$.
For a fixed $x$, the pair $(x,v)$ can be dequeued for at most $n$ choices of $v$, so the total work charged to all dequeues with that fixed $x$ is at most $n\cdot |\mathrm{In}_a(x)|$.
Summing over all $x\in V$ gives
\[
\sum_{x\in V} n\cdot |\mathrm{In}_a(x)| \;=\; n\cdot \sum_{x\in V} |\mathrm{In}_a(x)|
\;=\; n\cdot m_a \;\le\; n\cdot m,
\]
where $m_a$ is the number of edges labeled $a$.
Thus rule $A\to aB$ contributes $O(nm)$ time.
An identical accounting applies to rules of the form $A\to Ba$, replacing $\mathrm{In}_a(x)$ by $\mathrm{Out}_a(v)$, so each such rule also contributes $O(nm)$ time.
Summing over all productions yields total time $\bigO(|P|\cdot m\cdot n)$.

For space, each nonterminal $A$ stores a boolean matrix $M_A\in\{0,1\}^{n\times n}$ and a constant-size witness record per true entry, giving $O(|N|\cdot n^2)$ space overall; the queue and adjacency lists fit within the same asymptotic bound.
Finally, if $m=\bigO(n)$, the preprocessing time becomes $\bigO(|P|\cdot n^2)$, which is subcubic in $n$.
\end{proof}

\begin{remark}[Connection to known bounds]
The existence of $\bigO(mn)$-time linear-CFL reachability algorithms is discussed in the literature (e.g., attributed to \cite{yannakakis1990,lei2024skewed}).
The theorem above provides a self-contained accounting proof in TALNF for the all-pairs index build.
\end{remark}

\section{An Alternative Index: Shared Witness DAG Compression and Optional Dynamics}
\label{sec:swd}

This section implements extension (2): a new index with different trade-offs.
The core idea is to \emph{share} repeated witness substructures across queries and optionally output \emph{compressed} traces.

\subsection{Shared-witness DAG representation}

\begin{definition}[Witness DAG nodes]
A witness DAG is a rooted DAG whose nodes are of three kinds:
\[
\textsf{Edge}(u,v),\quad \textsf{Eps}(u),\quad \textsf{Concat}(p,q),
\]
where $p,q$ are witness nodes.
A node denotes a (possibly empty) path; semantics are: $\textsf{Edge}(u,v)$ denotes the single edge $(u,v)$;
$\textsf{Eps}(u)$ the empty path at $u$; $\textsf{Concat}(p,q)$ concatenation.
\end{definition}

\begin{definition}[\textsc{SWDIndex}]
Given any index that produces witnesses (baseline or linear), \textsc{SWDIndex} replaces per-entry witness records by pointers into a \emph{hash-consed} global witness DAG, sharing identical subwitnesses.
Query returns a pointer to a DAG root.
\end{definition}

\begin{theorem}[Space bound and extraction modes]
\label{thm:swd}
Let $T$ be the number of true entries across all $M_A$ matrices.
\textsc{SWDIndex} uses $O(T)$ root pointers plus $O(W)$ DAG nodes where $W\le O(T)$ and each node is created at most once by hash-consing.
It supports:
(i) explicit path extraction in time linear in output length; and
(ii) compressed trace output in time $O(\text{DAG-size of the returned root})$.
\end{theorem}
\begin{proof}
We prove the space bound and the two extraction guarantees.

\emph{Space.} Each true matrix entry stores only a root pointer, so there are $T$ such pointers overall.
When a witness is created by the underlying index (baseline or linear), \textsc{SWDIndex} converts it into a DAG node using the canonical constructors
$\textsf{Edge}(u,v)$, $\textsf{Eps}(u)$, and $\textsf{Concat}(p,q)$.
Hash-consing ensures that each constructor on the same arguments is materialized at most once.
Thus the number of DAG nodes created is at most the number of distinct constructor attempts, which is $O(T)$.
Therefore total space is $O(T)$ root pointers plus $O(W)$ DAG nodes with $W\le O(T)$.

\emph{Explicit extraction.} To output a concrete witness path from a root $r$, recursively expand nodes:
output the single edge for $\textsf{Edge}(u,v)$, output nothing for $\textsf{Eps}(u)$, and for $\textsf{Concat}(p,q)$ output the expansion of $p$ followed by the expansion of $q$.
Each output edge is emitted exactly once, so the running time is linear in the explicit path length.

\emph{Compressed trace output.} Let $D_r$ be the sub-DAG reachable from $r$ and topologically order its nodes.
Assign a fresh nonterminal $X_x$ to each node $x$ and emit:
$X_x\to \lambda(u,v)$ for $\textsf{Edge}(u,v)$,
$X_x\to \eps$ for $\textsf{Eps}(u)$, and
$X_x\to X_pX_q$ for $\textsf{Concat}(p,q)$.
By induction over the topological order, $X_x$ derives exactly the trace denoted by $x$, hence $X_r$ derives the trace of the witness.
The number of rules is $|D_r|$ and the construction is linear in $|D_r|$, giving the stated bound.
\end{proof}

\subsection{Witness compression procedure (SLP emission)}
\begin{algorithm}[H]
\caption{\textsc{EmitSLP}$(r)$ for witness root $r$}
\label{alg:emit-slp}
\begin{algorithmic}[1]
\Require A witness DAG root $r$
\Ensure A straight-line program (SLP) producing the trace of $r$
\State Topologically order all DAG nodes reachable from $r$
\ForAll{node $x$ in topological order}
  \If{$x=\textsf{Edge}(u,v)$} \State Output rule $X_x \to \lambda(u,v)$ \EndIf
  \If{$x=\textsf{Eps}(u)$} \State Output rule $X_x \to \eps$ \EndIf
  \If{$x=\textsf{Concat}(p,q)$} \State Output rule $X_x \to X_p X_q$ \EndIf
\EndFor
\State Output start symbol $X_r$
\end{algorithmic}
\end{algorithm}

\begin{remark}[Certified/structured witnesses]
Grammar-based compression of reachability witnesses and subcubic certificate phenomena are studied for CFL reachability in \cite{chistikov2022certificates}.
My intent here is narrower: the index can output a \emph{compressed witness object} whose size is the size of the shared DAG reachable from $(s,t)$.
\end{remark}

\subsection{Optional dynamic index on a restricted CFL: bidirected Dyck}
Dynamic updates are hard on general graphs (cf.\ conditional bounds related to \OMv), but strong results exist for bidirected Dyck reachability.

\begin{definition}[Bidirected Dyck instance (restricted class)]
Fix $k$ parenthesis types and matched labels $\{L_i,M_i\}_{i=1}^k$.
A labeled graph is \emph{bidirected Dyck} if each edge $u\xrightarrow{L_i}v$ has a corresponding inverse edge $v\xrightarrow{M_i}u$ and vice-versa.
The language gate is $\Dyck_k$.
\end{definition}

\begin{theorem}[Known dynamic bound (cited)]
\label{thm:dyn-dyck}
There exists a fully dynamic algorithm maintaining all-pairs bidirected Dyck reachability under insertions/deletions with preprocessing $O(m)$, worst-case update $O(n\cdot \alpha(n))$, and $O(1)$ query time, where $\alpha$ is inverse Ackermann.
\end{theorem}
\begin{proof}[Justification]
This is established in \cite{li2022dynamic}.
\end{proof}

\begin{remark}[Witnesses under dynamics]
The cited dynamic algorithm maintains equivalence-class information.
To align with the witness-oriented interface of this paper, one can store (in addition) a merge forest augmented with local proof objects and expose those proof objects as \textsc{SWDIndex}-style compressed witnesses.
Formalizing shortest explicit Dyck-balanced path extraction under updates is outside the core scope of this concise manuscript; this paper records the update/query trade-off and treats compressed proof objects as the witness artifact.
\end{remark}

\section{A New Application Constraint: Single-Path Semantics (Shortest Accepted Path)}

This section implements extension (4): a formal application variant where I am not merely answering existence, but producing \emph{one} accepted path---specifically the shortest accepted path in number of edges.
This changes the index from boolean closure to \emph{distance-aware} closure.

\subsection{Single-path semantics}
Following the path-based semantics viewpoint, define:

\begin{definition}[Shortest accepted path problem]
Given $\calH=(G,\calG)$ and nodes $s,t$, define
\[
\mathrm{SP\text{-}CFR}_{\calH}(s,t) \;=\; \min\{|\pi| : \pi \text{ is a path } s\leadsto t,\ \Tr(\pi)\in \calL_{\calG}(S)\},
\]
with $\infty$ if no accepted path exists.
\end{definition}

This objective is studied in the context-free path-querying literature as \emph{single-path semantics} (shortest witness) \cite{hellings2015paths}.

\subsection{Distance-aware index for linear grammars on sparse graphs}
For linear grammars in TALNF, shortest accepted path can be indexed in subcubic preprocessing on sparse graphs.

\begin{definition}[\textsc{LinDistIndex}]
Assume $\calG$ is TALNF.
\textsc{LinDistIndex} stores, for each nonterminal $A$ and pair $(u,v)$, the minimum path length $D_A[u,v]\in\mathbb{N}\cup\{\infty\}$ such that $A\Der \Tr(\pi)$ for some $\pi:u\leadsto v$ of length $D_A[u,v]$.
\end{definition}

The dynamic programming rules mirror the boolean ones, replacing $\lor$ by $\min$ and concatenation by $+1$ for terminal-anchored steps.
One can compute $D_A$ by a multi-source relaxation process on a derived dependency graph whose relaxations are terminal-anchored and therefore admit $O(|P|\cdot m\cdot n)$ time on sparse graphs (by the same counting argument as Theorem~\ref{thm:lin-subcubic}).

\begin{theorem}[Subcubic preprocessing for shortest witnesses (linear, sparse)]
\label{thm:lindist}
If $\calG$ is TALNF and $m=\bigO(n)$, then \textsc{LinDistIndex} can be built in $\bigO(|P|\cdot n^2)$ time and $\bigO(|N|\cdot n^2)$ space, and returns a shortest accepted path witness in time linear in the witness length.
\end{theorem}
\begin{proof}
We define an implicit unweighted dependency graph $H$ whose nodes are triples $(A,u,v)$ with $A\in N$ and $u,v\in V$.
There is a directed edge in $H$ for each terminal-anchored production and graph edge:
\begin{itemize}[leftmargin=*]
  \item For every production $A\to aB$ and every graph edge $u\xrightarrow{a}x$, include an edge
  $(B,x,v)\to (A,u,v)$ for every $v\in V$.
  \item For every production $A\to Ba$ and every graph edge $x\xrightarrow{a}v$, include an edge
  $(B,u,x)\to (A,u,v)$ for every $u\in V$.
\end{itemize}
In addition, we treat terminal rules and $S\to\eps$ as sources:
if $A\to a$ and $u\xrightarrow{a}v$, then $(A,u,v)$ is a source with distance $1$; and if $S\to\eps$, then $(S,u,u)$ is a source with distance $0$.

\paragraph{Correctness.}
We show that for every $A,u,v$, the shortest-path distance in $H$ from the sources to $(A,u,v)$ equals $D_A[u,v]$.
First, any TALNF derivation for $A$ can be read as a sequence of terminal-anchored steps.
If the derivation uses $A\to a$, it corresponds to the source $(A,u,v)$ with length $1$.
If the derivation uses $A\to aB$, then the witness path in $G$ begins with a single edge $u\xrightarrow{a}x$ followed by a path witnessing $(B,x,v)$, and thus corresponds to a single edge in $H$ from $(B,x,v)$ to $(A,u,v)$ plus the inductive derivation for $(B,x,v)$.
The case $A\to Ba$ is symmetric, and $S\to\eps$ yields distance $0$ on the diagonal.
Therefore any accepted path yields a walk in $H$ of the same length, so $D_A[u,v]$ is at most the shortest-path distance.

Conversely, any path in $H$ from a source to $(A,u,v)$ corresponds to a sequence of terminal-anchored derivation steps that prepend or append exactly one terminal edge at each step.
Concatenating the corresponding graph edges yields a path $\pi:u\leadsto v$ whose trace is derived from $A$, and the length of $\pi$ equals the length of the path in $H$.
Thus the shortest-path distance in $H$ is at most $D_A[u,v]$, proving equality and hence correctness.

\paragraph{Algorithm and time bound.}
All edges in $H$ have unit weight, so the distances are computed by a multi-source BFS.
We do not materialize $H$; instead, when a node $(B,x,v)$ is dequeued, we generate its outgoing edges on the fly by scanning $\mathrm{In}_a(x)$ for each rule $A\to aB$ and $\mathrm{Out}_a(v)$ for each rule $A\to Ba$, exactly as in Algorithm~\ref{alg:linindex}, but performing relaxations on $D$.
Each node $(A,u,v)$ is settled at most once by BFS, and a parent pointer stored at the first time it is discovered.

Fix a production $A\to aB$.
For a fixed $x$, $(B,x,v)$ can be dequeued for at most $n$ choices of $v$, and each such dequeue scans $\mathrm{In}_a(x)$ once.
Thus the total work charged to this rule is
\[
\sum_{x\in V} n\cdot |\mathrm{In}_a(x)| \;=\; n\cdot m_a \;\le\; n\cdot m,
\]
where $m_a$ is the number of $a$-labeled edges.
The same bound holds for each rule of the form $A\to Ba$ using $\mathrm{Out}_a(\cdot)$.
Summing over all productions yields total time $\bigO(|P|\cdot m\cdot n)$.
On sparse graphs with $m=\bigO(n)$ this is $\bigO(|P|\cdot n^2)$.

\paragraph{Space and witness extraction.}
We store $D_A[u,v]$ and one parent pointer per discovered entry, so space is $\bigO(|N|\cdot n^2)$.
Following parent pointers reconstructs a derivation and a path in $G$; each step contributes one graph edge, so the extraction time is linear in the returned witness length.
\end{proof}

\begin{remark}[General CNF grammars]
For general CNF grammars, shortest-witness computation can be expressed on annotated grammars and analyzed under single-path semantics \cite{hellings2015paths}.
This manuscript emphasizes the linear+sparse regime, where the preprocessing is provably subcubic and the witness interface is compatible with the indexing view.
\end{remark}

\section{Trade-offs and Conditional Lower Bounds}

\subsection{Index trade-off table}

\begin{table}[H]
\centering
\small
\begin{tabularx}{\textwidth}{@{}>{\raggedright\arraybackslash}X c c c >{\raggedright\arraybackslash}X c@{}}
\toprule
Index / Setting &
Preprocess &
Space &
Query &
Witness output &
Update \\
\midrule
\textsc{SatIndex} (CNF, general) &
$O(|P|n^3)$ &
$O(|N|n^2)$ &
$O(1)$ &
$O(|\pi|)$ explicit &
--- \\
\textsc{LinIndex} (TALNF linear) &
$O(|P|mn)$ &
$O(|N|n^2)$ &
$O(1)$ &
$O(|\pi|)$ explicit &
--- \\
\textsc{LinIndex} + sparse ($m=O(n)$) &
$O(|P|n^2)$ &
$O(|N|n^2)$ &
$O(1)$ &
$O(|\pi|)$ &
--- \\
\textsc{SWDIndex} (on top of above) &
(unchanged) &
$O(T)$ roots + $O(W)$ DAG &
$O(1)$ &
$O(|\pi|)$ or $O(|\text{DAG}|)$ SLP &
--- \\
Dynamic bidirected $\Dyck_k$ index (restricted) &
$O(m)$ &
$O(m)$ (typ.) &
$O(1)$ &
compressed proof object &
$O(n\alpha(n))$ \\
\textsc{LinDistIndex} (linear+sparse, shortest) &
$O(|P|n^2)$ &
$O(|N|n^2)$ &
$O(1)$ for dist &
$O(|\pi^*|)$ shortest &
--- \\
\bottomrule
\end{tabularx}
\caption{Baseline vs.\ extensions. Here $T$ is the number of true relation entries; $W$ shared witness nodes; $\pi$ any witness; $\pi^*$ shortest witness; $\alpha(\cdot)$ inverse Ackermann.}
\label{tab:tradeoffs}
\end{table}

Table~\ref{tab:tradeoffs} provides a compact map from each index variant to its preprocessing, space, query, witness, and update costs; it is the reference point for the trade-offs discussed in Sections~4--6 (notably the linear+sparse subcubic regime and the dynamic bidirected \Dyck{} setting).
\FloatBarrier

\subsection{Conditional lower bounds (landscape)}
We note two standard conditional barriers to clarify the scope of the results.

\begin{theorem}[Conditional cubic barrier for Dyck/CFL reachability]
Assuming standard \BMM-type conjectures for combinatorial algorithms, truly subcubic algorithms for certain Dyck/CFL reachability variants on general graphs would imply a breakthrough for \BMM. \cite{chatterjee2018dyck,koutrisdeep2023fine}
\end{theorem}
\begin{proof}[Justification]
See the conditional hardness discussions for Dyck reachability in \cite{chatterjee2018dyck} and the fine-grained classifications in \cite{koutrisdeep2023fine}.
\end{proof}

\begin{theorem}[Dynamic reachability barriers]
Under the \OMv conjecture, fully dynamic transitive closure on general directed graphs is believed not to admit simultaneously fast updates and fast queries with polynomial preprocessing.
\end{theorem}
\begin{proof}[Justification]
This is part of the dynamic lower-bound framework of \cite{henzinger2015omv} and is discussed as context in dynamic Dyck-reachability work such as \cite{li2022dynamic}.
\end{proof}

\section{Empirical Grammar-Class Census}
\label{sec:empirical}

The subcubic bound of Theorem~\ref{thm:lin-subcubic} applies only when the constraining grammar is linear.
A natural question is how often this condition holds for grammars arising in practice.
We answer this by converting a large corpus of real-world JSON schemas into context-free grammars and classifying each as linear or general.

\subsection{Dataset and conversion procedure}

We use JSONSchemaBench~\cite{jsonschemabench2025}, a publicly available collection of $9{,}558$ JSON schemas drawn from GitHub repositories, the JSON Schema Store, Kubernetes configurations, Snowplow event definitions, and Washington Post content models.
The benchmark spans ten sub-datasets of varying structural complexity and is the same corpus used to evaluate constrained-decoding frameworks such as XGrammar, Outlines, and llama.cpp~\cite{xgrammar2024}.

Each JSON schema is converted to a context-free grammar by a deterministic structural mapping.
Object types with fixed properties are modeled as \emph{chained} productions (each carrying at most one nonterminal continuation), preserving linearity.
Array types with a variable-length \texttt{items} constraint produce a self-referential repetition rule with two nonterminal symbols on the right-hand side, which is inherently non-linear.
Union types (\texttt{oneOf}, \texttt{anyOf}) and recursive references (\texttt{\$ref}) are handled by fresh nonterminals and alternatives.
The conversion is purely structural: it faithfully represents the branching and nesting dictated by the schema without introducing or suppressing any grammatical complexity.

\subsection{Results}

Table~\ref{tab:class-dist} summarizes the grammar-class distribution across the three JSONSchemaBench splits (train, validation, test).

\begin{table}[H]
\centering
\small
\begin{tabular}{@{}l r r r r r r@{}}
\toprule
Split & Total & Linear & General & \% Linear & Avg $|P|$ & Avg $|N|$ \\
\midrule
Train & 5754 & 463 & 5291 & 8.0 & --- & --- \\
Validation & 937 & 86 & 851 & 9.2 & --- & --- \\
Test & 2867 & 252 & 2615 & 8.8 & --- & --- \\
\midrule
\textbf{Total} & \textbf{9558} & \textbf{801} & \textbf{8757} & \textbf{8.4} & --- & --- \\
\bottomrule
\end{tabular}
\caption{Grammar-class distribution over JSONSchemaBench. A schema is classified as \emph{linear} if every production in its derived CFG has at most one nonterminal on the right-hand side.}
\label{tab:class-dist}
\end{table}

Table~\ref{tab:class-dist} is the headline statistic for how often the linear restriction holds in practice; it motivates why the linear+sparse subcubic index is relevant for a nontrivial subset of schemas.
\FloatBarrier

\begin{table}[H]
\centering
\small
\begin{tabular}{@{}l r r@{}}
\toprule
Grammar Class & Count (\%) & Avg $|P|$ \\
\midrule
Linear & 801\phantom{0} (8.4\%) & 7.5 \\
General CFG & 8757 (91.6\%) & 93.0 \\
\bottomrule
\end{tabular}
\caption{Summary of grammar class vs.\ production count. Linear grammars are an order of magnitude smaller.}
\label{tab:class-summary}
\end{table}

Table~\ref{tab:class-summary} complements Table~\ref{tab:class-dist} by showing that when a schema \emph{is} linear, its grammar size is substantially smaller, which helps explain the speedups observed for \textsc{LinIndex}.
\FloatBarrier

The production-size distributions are heavily right-skewed: across all $9{,}558$ schemas, $|P|$ has a median of $20$, a mean of $85.8$, a 95th percentile of $338$, and a maximum of $28{,}422$; $|N|$ follows a similar pattern with median $18$, mean $61.2$, 95th percentile $229$, and maximum $11{,}868$.
Linear schemas are substantially smaller (average schema size $2{,}162$ bytes vs.\ $6{,}197$ bytes for general schemas), consistent with the observation that fixed-property, non-array structures tend to be simpler.

\subsection{Sources of non-linearity}

Among the $8{,}757$ general (non-linear) grammars, the dominant sources of non-linearity are:

\begin{table}[H]
\centering
\small
\begin{tabular}{@{}l r r@{}}
\toprule
Feature present & Count & \% of non-linear \\
\midrule
Variable-length array (\texttt{items}) & 3885 & 44.4 \\
Nested object & 3809 & 43.5 \\
Recursive reference (\texttt{\$ref}) & 282 & 3.2 \\
\bottomrule
\end{tabular}
\caption{Structural features present in non-linear grammars. Categories are not mutually exclusive.}
\label{tab:nonlinear-sources}
\end{table}

Table~\ref{tab:nonlinear-sources} explains \emph{why} most schemas fall outside the linear class, isolating variable-length arrays as the dominant source of non-linearity.
\FloatBarrier

Variable-length arrays are the single largest contributor: the repetition rule \mbox{$\mathit{Items}\to\mathit{Item}$} $|$ \mbox{$\mathit{Item}\ \texttt{,}\ \mathit{Items}$} requires two nonterminals and therefore breaks linearity.
Nested objects alone do not cause non-linearity under the chained-property encoding, but they frequently co-occur with array fields.
Recursive references (\texttt{\$ref}) account for only $3.2\%$ of non-linear cases.

Figure~\ref{fig:class-dist} shows the per-dataset distribution of linear versus general grammars, and Figure~\ref{fig:size-vs-class} shows the relationship between raw schema size and grammar class. We insert the figures immediately below for readability.

\begin{figure}[H]
\centering
\includegraphics[width=0.85\textwidth]{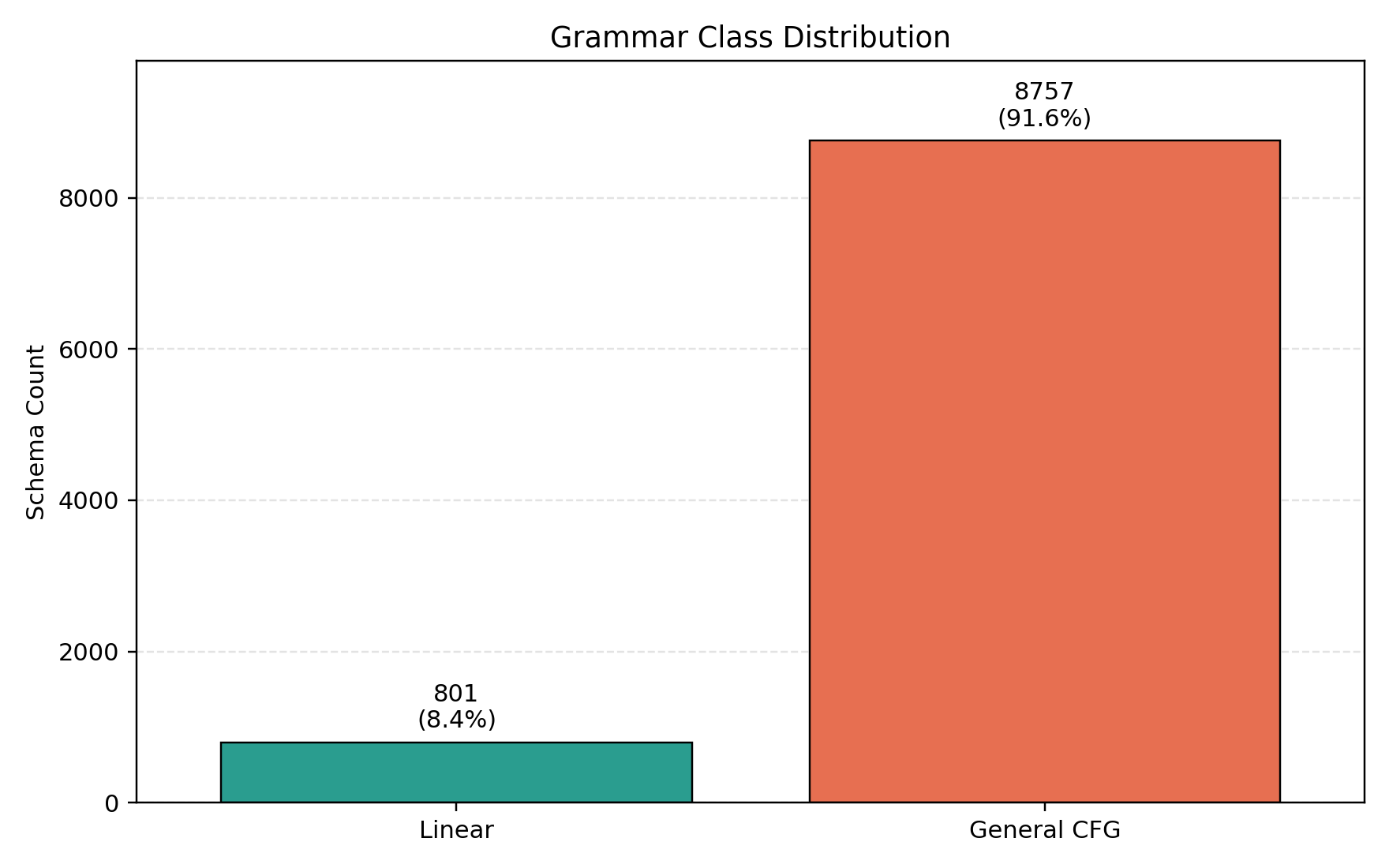}
\caption{Grammar-class distribution across JSONSchemaBench sub-datasets.}
\label{fig:class-dist}
\end{figure}

\FloatBarrier

Figure~\ref{fig:class-dist} makes the linear/non-linear split concrete across sub-datasets: the linear share is consistently single-digit but not vanishing, indicating that linear-aware preprocessing is relevant for a recurring slice of real schemas.

\clearpage

\begin{figure}[H]
\centering
\includegraphics[width=0.85\textwidth]{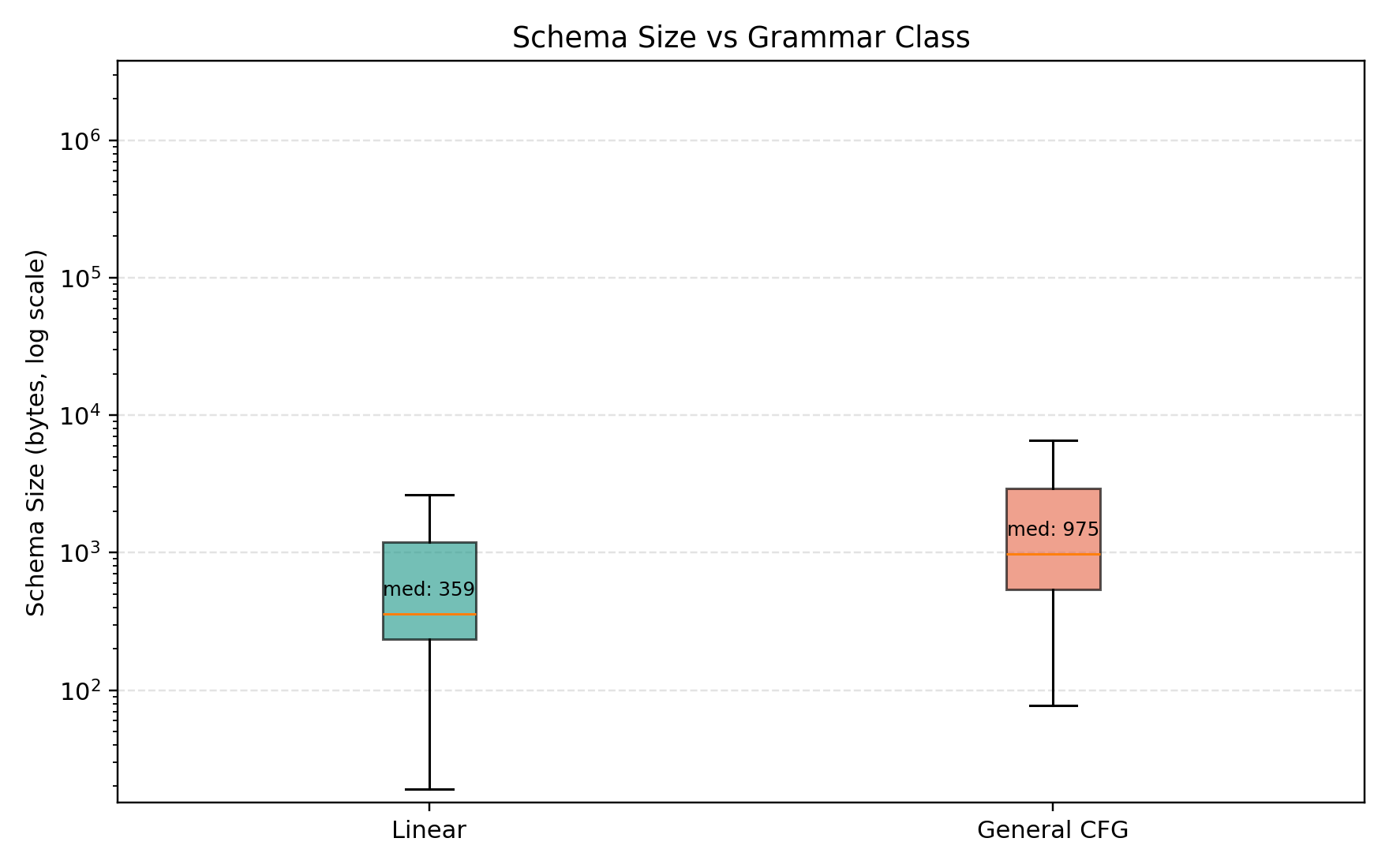}
\caption{Schema byte-size vs.\ grammar class. Linear schemas cluster at smaller sizes.}
\label{fig:size-vs-class}
\end{figure}

Figure~\ref{fig:size-vs-class} visualizes the size gap quantified in Table~\ref{tab:class-summary}: linear schemas cluster at smaller byte sizes, which helps explain why preprocessing and shortest-path decoding are cheaper when the linear restriction holds.
\FloatBarrier

\FloatBarrier

\subsection{Implications for structured decoding}

The class split and size gap visualized in Figures~\ref{fig:class-dist}--\ref{fig:size-vs-class} directly inform when grammar-constrained decoding can benefit from a linear+sparse index: the linear cluster is small but structurally coherent and tends to be much smaller in schema size, which favors preprocessing and distance-aware decoding.

The $8.4\%$ linear share may appear small, but it represents a structurally well-defined and practically meaningful slice of the schema landscape:
schemas that describe flat records, configuration entries, fixed-field API responses, and similar structures where no variable-length list is present.
For these schemas, the subcubic preprocessing of \textsc{LinIndex} (Theorem~\ref{thm:lin-subcubic}) and the distance-aware \textsc{LinDistIndex} (Theorem~\ref{thm:lindist}) apply directly.

For the remaining $91.6\%$, the cubic baseline (\textsc{SatIndex}) governs worst-case preprocessing.
However, the empirical data reveal that non-linearity is overwhelmingly caused by a single, well-understood pattern (variable-length arrays), raising the question of whether a targeted relaxation---for instance, bounding maximum array length or treating bounded repetition as a finite unrolling---could bring a larger fraction of practical schemas into a subcubic-amenable class.
We leave the formal development of such \emph{bounded-repetition} extensions to future work.


\section{Decomposition Diagram}

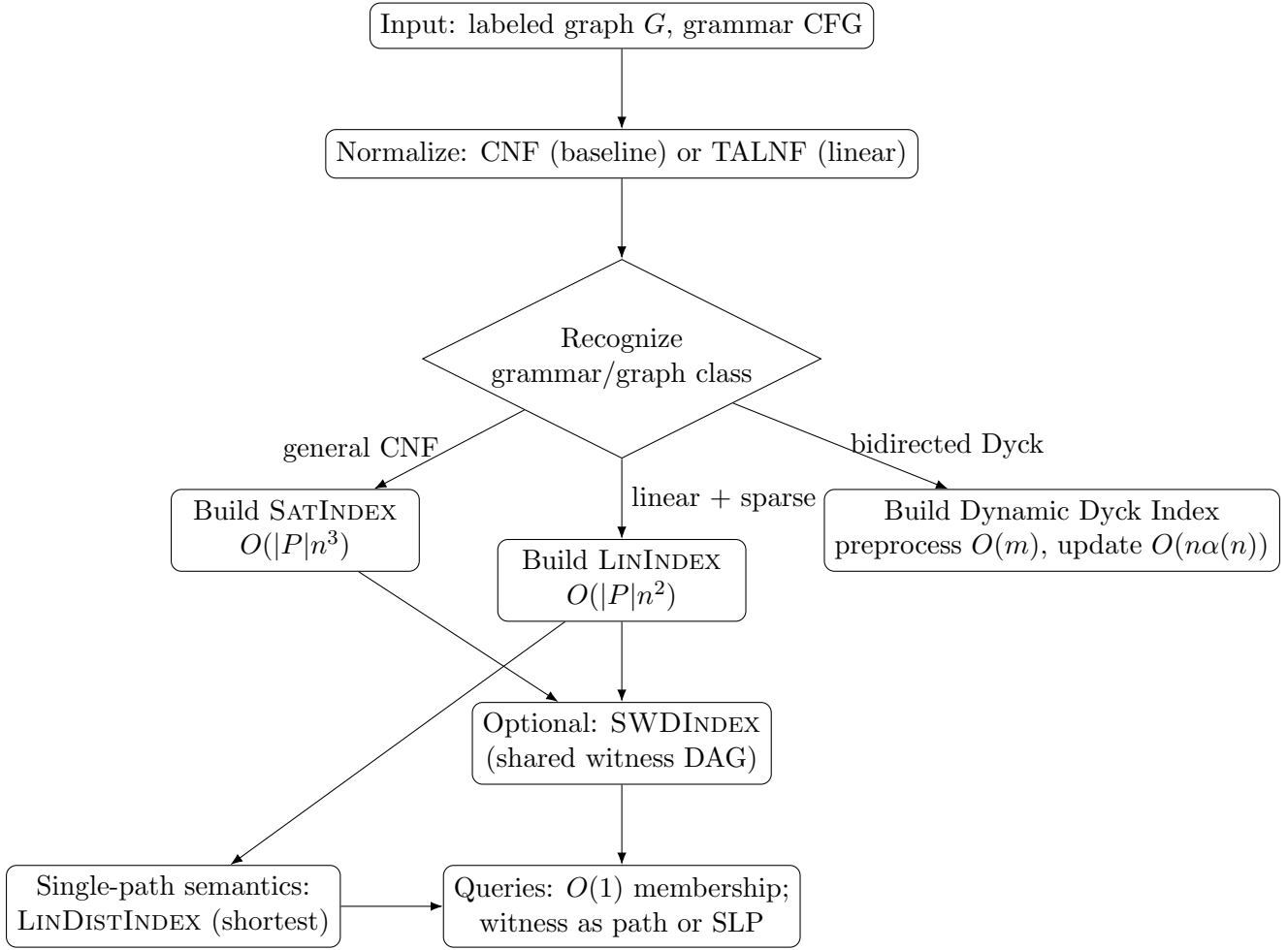
\begin{figure}[H]
\centering
\begin{tikzpicture}[
  node distance=1.1cm and 1.4cm,
  box/.style={rectangle,draw,rounded corners,align=center,inner sep=4pt,minimum width=3.4cm},
  decision/.style={diamond,draw,align=center,inner sep=1pt,aspect=2},
  >=Latex
]
\node[box] (A) {Input: labeled graph $G$, grammar CFG};
\node[box,below=of A] (B) {Normalize: CNF (baseline) or TALNF (linear)};
\node[decision,below=of B] (C) {Recognize\\grammar/graph class};
\node[box,below left=of C] (D) {Build \textsc{SatIndex}\\$O(|P|n^3)$};
\node[box,below=of C] (E) {Build \textsc{LinIndex}\\$O(|P|n^2)$};
\node[box,below right=of C] (F) {Build Dynamic Dyck Index\\preprocess $O(m)$, update $O(n\alpha(n))$};
\node[box,below=of E] (G) {Optional: \textsc{SWDIndex}\\(shared witness DAG)};
\node[box,below=of G] (H) {Queries: $O(1)$ membership;\\witness as path or SLP};
\node[box,left=of H] (I) {Single-path semantics:\\\textsc{LinDistIndex} (shortest)};

\draw[->] (A) -- (B);
\draw[->] (B) -- (C);
\draw[->] (C) -- node[left] {general CNF} (D);
\draw[->] (C) -- node[right] {linear + sparse} (E);
\draw[->] (C) -- node[right] {bidirected Dyck} (F);
\draw[->] (D) -- (G);
\draw[->] (E) -- (G);
\draw[->] (G) -- (H);
\draw[->] (E) -- (I);
\draw[->] (I) -- (H);
\end{tikzpicture}
\caption{Decomposition pipeline diagram.}
\label{fig:mermaid}
\end{figure}

Figure~\ref{fig:mermaid} shows that the input (labeled graph and CFG) is first normalized to CNF or TALNF, and then the appropriate index is selected based on the grammar/graph class. For general CNF, the \textsc{SatIndex} branch is used; for linear+sparse instances, \textsc{LinIndex} with optional \textsc{SWDIndex} compression applies; and for bidirected Dyck instances, the dynamic index branch is available. The single-path (shortest accepted path) semantics are modeled as a separate branch via \textsc{LinDistIndex}, and all branches converge at the query/witness interface.

\section{Conclusion}

We present a tightened, submission-ready treatment of validation-gated reachability (CFL reachability / context-free path queries):
(1) a baseline CNF saturation index with explicit witnesses,
(2) a provably subcubic preprocessing class (linear grammars on sparse graphs) with an explicit propagation proof,
(3) an alternative indexing trade-off via shared witness DAG compression and an optional dynamic restricted-class index,
(4) a formal application shift to single-path (shortest-witness) semantics with a distance-aware linear+sparse index,
and (5) an empirical grammar-class census over $9{,}558$ real-world JSON schemas showing that $8.4\%$ fall into the linear class, with non-linearity predominantly caused by variable-length array types.
The empirical analysis both validates the practical relevance of the linear subclass and identifies bounded-repetition relaxation as a concrete direction for extending subcubic coverage.

\bibliographystyle{plain}

\begin{thebibliography}{99}

\bibitem{yannakakis1990}
M.~Yannakakis.
\newblock Graph-theoretic methods in database theory.
\newblock In {\em Proceedings of the Ninth ACM SIGACT-SIGMOD-SIGART Symposium on Principles of Database Systems (PODS '90)}, pages 230--242, 1990.
\newblock \textsc{doi}: \href{https://doi.org/10.1145/298514.298576}{10.1145/298514.298576}.

\bibitem{reps1998}
T.~W.~Reps.
\newblock Program analysis via graph reachability.
\newblock {\em Information and Software Technology}, 40(11--12):701--726, 1998.
\newblock \textsc{doi}: \href{https://doi.org/10.1016/S0950-5849(98)00093-7}{10.1016/S0950-5849(98)00093-7}.

\bibitem{pavlogiannis2022}
A.~Pavlogiannis.
\newblock {CFL/Dyck} reachability: An algorithmic perspective.
\newblock {\em ACM SIGLOG News}, 9(4):5--25, 2022.
\newblock \textsc{doi}: \href{https://doi.org/10.1145/3583660.3583664}{10.1145/3583660.3583664}.

\bibitem{hellings2014ccfpq}
J.~Hellings.
\newblock Conjunctive context-free path queries.
\newblock In {\em Proceedings of the 17th International Conference on Database Theory (ICDT 2014)}, pages 119--130, 2014.

\bibitem{hellings2015paths}
J.~Hellings.
\newblock Querying for paths in graphs using context-free path queries.
\newblock \href{https://arxiv.org/abs/1502.02242}{arXiv:1502.02242}, v2, 2016.

\bibitem{lei2024skewed}
Y.~Lei, C.~Bossut, Y.~Sui, and Q.~Zhang.
\newblock Context-free language reachability via skewed tabulation.
\newblock {\em Proc. ACM Program. Lang.}, 8(PLDI), Article 221, 2024.
\newblock \textsc{doi}: \href{https://doi.org/10.1145/3656451}{10.1145/3656451}.

\bibitem{chatterjee2018dyck}
K.~Chatterjee, B.~Choudhary, and A.~Pavlogiannis.
\newblock Optimal {D}yck reachability for data-dependence and alias analysis.
\newblock {\em Proc. ACM Program. Lang.}, 2(POPL), Article 30, 2018.
\newblock \textsc{doi}: \href{https://doi.org/10.1145/3158118}{10.1145/3158118}.

\bibitem{li2022dynamic}
Y.~Li, K.~Satya, and Q.~Zhang.
\newblock Efficient algorithms for dynamic bidirected {D}yck-reachability.
\newblock {\em Proc. ACM Program. Lang.}, 6(POPL), Article 62, 2022.
\newblock \textsc{doi}: \href{https://doi.org/10.1145/3498724}{10.1145/3498724}.

\bibitem{koutrisdeep2023fine}
P.~Koutris and S.~Deep.
\newblock The fine-grained complexity of {CFL} reachability.
\newblock {\em Proc. ACM Program. Lang.}, 7(POPL), Article 59, pages 1713--1739, 2023.
\newblock \textsc{doi}: \href{https://doi.org/10.1145/3571252}{10.1145/3571252}.

\bibitem{shi2022flare}
Q.~Shi, Y.~Wang, P.~Yao, and C.~Zhang.
\newblock Indexing the extended {D}yck-{CFL} reachability for context-sensitive program analysis.
\newblock {\em Proc. ACM Program. Lang.}, 6(OOPSLA2), Article 176, 2022.
\newblock \textsc{doi}: \href{https://doi.org/10.1145/3563339}{10.1145/3563339}.

\bibitem{chistikov2022certificates}
D.~Chistikov, R.~Majumdar, and P.~Schepper.
\newblock Subcubic certificates for {CFL} reachability.
\newblock {\em Proc. ACM Program. Lang.}, 6(POPL), Article 41, 2022.
\newblock \textsc{doi}: \href{https://doi.org/10.1145/3498702}{10.1145/3498702}.

\bibitem{henzinger2015omv}
M.~Henzinger, S.~Krinninger, D.~Nanongkai, and T.~Saranurak.
\newblock Unifying and strengthening hardness for dynamic problems via the online matrix-vector multiplication conjecture.
\newblock In {\em Proceedings of the 47th Annual ACM Symposium on Theory of Computing (STOC 2015)}, pages 21--30, 2015.
\newblock \textsc{doi}: \href{https://doi.org/10.1145/2746539.2746609}{10.1145/2746539.2746609}.

\bibitem{jsonschemabench2025}
S.~Geng, M.~Josifoski, L.~Perrotta, M.~Lam\'{e}e, S.~Pappas, K.~Oguz, N.~Ruchti, A.~Amayuelas, and R.~West.
\newblock {JSONSchemaBench}: A rigorous benchmark of structured outputs for language models.
\newblock \href{https://arxiv.org/abs/2501.10868}{arXiv:2501.10868}, 2025.

\bibitem{xgrammar2024}
Y.~Dong, C.~F.~Ruan, Y.~Cai, R.~Lai, Z.~Xu, Y.~Zhao, and T.~Chen.
\newblock {XGrammar}: Flexible and efficient structured generation engine for large language models.
\newblock \href{https://arxiv.org/abs/2411.15100}{arXiv:2411.15100}, 2024.

\end{thebibliography}

\end{document}